\documentclass[letter,aps,pra,twocolumn]{revtex4-2}
\usepackage[utf8]{inputenc}
\usepackage{bigints}
\usepackage{amsfonts}
\usepackage{dsfont}
\usepackage{latexsym}
\usepackage{amsmath, amsthm, amssymb, mathtools}
\DeclareMathOperator{\Tr}{Tr}
\usepackage{enumerate}
\usepackage{fancybox}
\usepackage{graphicx}
\usepackage{xcolor}
\usepackage[amssymb, cdot]{SIunits}
\usepackage[shortlabels]{enumitem}
\usepackage[toc]{appendix}
\usepackage{calc}
\newlength\dlf

\theoremstyle{definition}
\newtheorem{definition}{Definition}[section]

\newtheorem{theorem}{Theorem}[section]

\newtheorem{corollary}{Corollary}[theorem]

\theoremstyle{remark}

\usepackage{hyperref}

\begin{document}

\title{\textbf{Existence of a robust optimal control process for efficient measurements in a two-qubit system } 
}%

\author{Ricardo Rodriguez}
    \email[Correspondence email address: ]{rodrig.ricardo@gmail.com}
    \affiliation{Department of Natural Sciences and Mathematics, McPherson College, McPherson, KS, USA}

\author{Nam Nguyen}
    \affiliation{Applied Mathematics, Boeing Research \& Technology, Huntington Beach, CA, USA}

\author{Elizabeth Behrman}
    \affiliation{Department of Mathematics, Statistics, and Physics, Wichita State University, Wichita, KS, USA}

\author{Andy C.~Y.~Li}
    \affiliation{Emerging Technologies Directorate, Fermi National Accelerator Laboratory, Batavia, IL, USA}

\author{James Steck}
    \affiliation{Department of Aerospace Engineering, Wichita State University, Wichita, KS, USA}

\date{\today}

\begin{abstract}
The verification of quantum entanglement is essential for quality control in quantum communication. In this work, we propose an efficient protocol to directly verify the two-qubit entanglement of a known target state through a single expectation value measurement. Our method provides exact entanglement quantification using the concurrence measure without performing quantum state tomography. We prove the existence of a unitary transformation that drives a known initial state of a two-qubit system to a designated final state, where the trace over a chosen observable directly yields the concurrence of the initial state. Furthermore, we implement an optimal control process of that transformation and demonstrate its effectiveness through numerical simulations. We also show that this process is robust to environmental noise. Our approach offers advantages in directly verifying entanglement with low circuit depth, making it suitable for industrial-scale quality control of entanglement generation. Our results, presented here, provide mathematical justification for our earlier computational experiments.
\end{abstract}

\maketitle



\section{Introduction}

Recent developments in quantum information science have shown progress toward practical quantum advantage \cite{Google2024,Drmota2025,Smith2025,Kim2025} over classical information theory, as formalized by Claude Shannon \cite{Shannon1948}. Quantum entanglement is one of the most crucial features distinguishing quantum information from its classical counterpart, and it is the fundamental building block of quantum communication \cite{rohde2025quantuminternettechnicalversion} and quantum computing \cite{nielsen00}. In particular, entanglement serves as an important resource in quantum cryptography and its various manifestations \cite{Mastriani:2023vza}. For example, quantum key distribution \cite{Bennett_2014}
quantum secure direct communications \cite{PhysRevLett.68.3121},
quantum secret sharing \cite{Joy_2019}, and quantum key secure communication \cite{Mastriani:2022uap} all depend, in a large part, on the sharing of entangled particles. Another novel application involving entanglement is the production of traceable random numbers \cite{kavuri2024traceablerandomnumbersnonlocal}.

Any practical applications of quantum communications and the future quantum internet will require an industrial production of entangled systems and their quality control. Verifying the two-qubit entanglement generated by the production protocol plays a critical role in quality control, and the optimal entanglement verification has attracted considerable interest in recent years \cite{Pallister2018,Goswami2019,Wang2019,Riccardi2021,Riera2023,Chen2025}.
One direct way to measure entanglement is by quantum state tomography \cite{Steffen2006}, but it is not the most efficient protocol since it measures the complete state information. More efficient protocols have been explored by designing approximate entanglement measures \cite{Riccardi2021,Riera2023} and utilizing quantum memory \cite{Chen2025}.

We envision the following scenario for our results. A source of pairs with known entanglement level is used in a particular application, and the level of entanglement needs to be certified in order to ensure the reliability of such use. We propose that this can be done by sampling enough pairs, and using a known control process to drive these pairs to a final state where the entanglement can be verified by only a single expectation value whose numerical outcome should be equal to the known concurrence \cite{Wootters:1997id} of the initial pair. If the initial state has a different entanglement value, the control process will produce a different value from the expected one. This is more efficient than other methods, like full tomography, since it requires fewer measurements.  Our protocol is also robust to constant parameter drift, a crucial consideration for practical applications.
We also demonstrate numerically that an optimal control process can efficiently implement the protocol.

The remainder of this paper is organized as follows. Section \ref{Sec.Math.Fram} describes the mathematical framework used to prove the existence of a unitary transformation from an arbitrary state to a designated state. The trace of this state over a chosen observable provides the concurrence (equivalently, the entanglement) of the initial state. In Section \ref{Sec.Opt.Cont} we define a control functional that allows us to find the final state and a control path. In Section \ref{Sec.Quant.Cont.Alg} we explain how the above control process can be implemented as a programmable algorithm. Our numerical implementation is presented and discussed in Section \ref{Sec.Dis.Resul}, followed by a brief Conclusion in Section \ref{thatsallfolks}. The appendix contains several important mathematical definitions and facts.

\section{Mathematical Framework}\label{Sec.Math.Fram}

\subsection{Existence and robustness of the unitary transformations}
\label{subsec:existence_and_robustness}

The time evolution of a quantum system with density matrix $\rho$ and Hamiltonian $H$ is given by the Liouville-Von Neumann equation:
\begin{equation}
 \frac{\partial \rho}{\partial t}=\frac{1}{i\hbar}[H,\rho].\label{eq:Liouville}
\end{equation}
Generically, a Hamiltonian is described by a linear combination of terms that correspond to a certain physical set of interactions. For example:
\begin{align}
\label{eq:Hamiltonian}
H = &  \, \kappa_A \sigma_x\otimes\mathds{1}+\kappa_B \mathds{1}\otimes\sigma_x+\epsilon_A\sigma_z\otimes\mathds{1}+\epsilon_B \mathds{1}\otimes\sigma_z \nonumber \\
&+\zeta\sigma_z\otimes\sigma_z+\nu \sigma_x\otimes\sigma_x
\end{align}
is QMA-complete \cite{Biamonte:2008kyd} and a generalization of the Hamiltonian used by \cite{behrman2008}. $\kappa_i$ are the tunneling amplitudes, $\epsilon_i$ are the biases, and $\zeta$ and $\nu$ are qubit-qubit couplings. $H$ is a symmetric matrix if the parameters are real. We assume that $\kappa_i,\epsilon_i,\zeta,\nu$ are functions of time and will represent their time variations by the generic control process ${\mathbf u}(t)$. \\

The system eq.(\ref{eq:Liouville}) is density matrix controllable (DMC) \cite{Dalessandro_2022} if, for each pair of unitary equivalent matrices $\rho_i$ and $\rho_f$, there exists a control process $\mathbf{u}$ and a time $T>0$ such that the solution $\rho(t)$ of eq.(\ref{eq:Liouville}) with initial condition equal to $\rho_i$ at time $t=T$ is $\rho_f$. This definition means, in our particular case, that given an initial and final state, we can always find real functions $\mathbf{u}(t) =\{ \kappa_i(t),\epsilon_i(t),\zeta(t),\nu(t)\}$ that allow the Hamiltonian to evolve one state to the other. The natural question is to guarantee the existence of the process $\mathbf{u}(t)$.

\begin{definition}[Dynamical Lie Algebra \cite{Dalessandro_2022}] 
    Let a control Hamiltonian be described by 
    \begin{equation*}
        H(t) = \sum_{j=1}^m u_j(t) H_j,
    \end{equation*}
    where $\{H_j\}$ are time-independent basis operators (e.g. Pauli strings). The Dynamical Lie Algebra (DLA), denoted as $\mathcal{L}$, is the Lie algebra generated by the set of skew-Hermitian operators
    \begin{equation*}
        S = \{iH_1, iH_2, \cdots, iH_m \}.
    \end{equation*}
    This is the smallest Lie subalgebra of $\mathfrak{su}(n)$ that contains the set $S$, and is equivalent to the vector space spanned by all possible repeated Lie brackets of the elements of $S$.
\end{definition}

\begin{theorem}[Existence of the unitary transformation]
\label{thm:existence}
    The quantum control system for two-qubit system described by the Hamiltonian in Equation \ref{eq:Hamiltonian} is DMC.
\end{theorem}

\begin{proof}
    To show that this quantum system is DMC, we must show that its DLA is the special unitary algebra $\mathfrak{su}(4)$ according to the theory of quantum control \cite{Dalessandro_2022}. Note that for the Hamiltonian in Equation \ref{eq:Hamiltonian}, the set of six-skew Hermitian basis operators is 
    \begin{equation*}
        S = \{i\sigma_x\otimes\mathds{1},  i\mathds{1}\otimes\sigma_x, i\sigma_z\otimes\mathds{1}, i\mathds{1}\otimes\sigma_z, i\sigma_z\otimes\sigma_z, i\sigma_x\otimes\sigma_x\}
    \end{equation*}
    The DLA can be generated by repeatedly computing the Lie brackets (commutators). One can easily show that the first round of this calculation yields 7 new, linearly independent operators. The second round yields 2 more linearly independent operators. Thus, after two rounds, we have 15 linearly independent operators. Since the dimension of the Lie algebra of $\mathfrak{su}(4)$ is 15, we have shown that the DLA of the system is indeed  $\mathfrak{su}(4)$. Therefore, by the theory of quantum control, this system is DMC. 
\end{proof}

Having established the controllability (existence of the unitary transformation) in Theorem \ref{thm:existence}, we now investigate its robustness to common physical noise. In particular, the drift in the Hamiltonian.

\begin{corollary}[Robust controllability]
    \label{cor:robustness}
    The two-qubit control system governed by the ideal Hamiltonian in Equation \ref{eq:Hamiltonian} remains DMC in the presence of drift in the Hamiltonian.
\end{corollary}

\begin{proof}
    The total Hamiltonian of the imperfect (noisy) system is 
    \begin{equation*}
        H_{\text{total}} = H + H_d 
    \end{equation*}
    where $H$ is the ideal Hamiltonian, such as the Hamiltonian described in Equation \ref{eq:Hamiltonian}, and $H_d$ is the drift term and can be characterized as a general Hermitian matrix \cite{noise}. For a two-qubit system, the DLA, $\mathcal{L_{\text{total}}}$, is $\mathfrak{su}(4)$. The set of generators for this system is:
    \begin{equation*}
        G_{\text{total}} = S \cup \{iH_d\}
    \end{equation*}
    where $S$ is the set of the generators from the ideal Hamiltonian. We have established in Theorem \ref{thm:existence} that the Lie algebra generated by the set $S$ from the Hamiltonian in Equation \ref{eq:Hamiltonian} is the entire algebra. That is, $\mathcal{L_{H}} = \mathfrak{su}(4)$. The skew-Hermitian operator $\{i H_d\}$ is an element of the algebra $\mathfrak{u}(4)$, and its traceless part lies within $ \mathfrak{su}(4)$. 

    Now, the DLA of the total system, $\mathcal{L}_{\text{total}} = \text{Lie}(G_{\text{total}})$, is  generated by a set of operators that contains the generated set of $\mathcal{L_{H}}$. Therefore,
    \begin{equation*}
        \mathfrak{su}(4) = \mathcal{L_{H}} \subseteq \mathcal{L}_{\text{total}}
    \end{equation*}
    Since $\mathfrak{su}(4)$ is the largest possible DLA for a two-qubit system, we must have  $\mathcal{L}_{\text{total}} =\mathfrak{su}(4). $
\end{proof}

Corollary \ref{cor:robustness} shows that the vector space generated by the controls in Equation \ref{eq:Hamiltonian} is large enough to allow for active, real-time compensation for any unwanted drift. The proof demonstrates that for any drift induced by $H_d$, there exists a set of controls $\mathbf{u}(t)$ that can simultaneously counteract the drift and  steer the system toward any desired state.

\subsection{Existence of the Target State}
\label{sec:existence_of_target_state}
The preceding analysis established that the control Hamiltonian (Equation \ref{eq:Hamiltonian}) is universal. This guarantees that any unitary $U$ can be implemented, hence, the system can be driven from an initial state $\rho(0)$ to any unitarily equivalent state. The specific goal of quantum control is to find a state within this set of unitarily equivalent states that yields a particular measurement outcome. The existence of such a state is not guaranteed. To prove this, we recall the fundamental link between unitary equivalence and the spectrum of a matrix.

\begin{theorem} [Unitary equivalence for two Hermitian matrices \cite{Horn_1985}]
Two Hermitian matrices,  $A$ and $B$, are unitarily equivalent (i.e., $A = UBU^\dagger$ for some unitary $U$) if and only if they share the same spectrum of eigenvalues.
\end{theorem}

This theorem allows us to define the set of all possible final states as the collection of all states unitarily equivalent to the initial state. For a given initial spectrum $\lambda$, the following theorem proves that this set is guaranteed to contain a state that satisfies our measurement condition. 

\begin{theorem}
    For any initial two-qubit state $\rho(0)$, there exists a final state $\rho(t_f)$ that is unitarily equivalent to $\rho(0)$ and satisfies the condition $\Tr((\sigma_z \otimes \sigma_z) \cdot \rho(t_f)) = \mathcal{C}$, where $\mathcal{C}$ is the concurrence of the initial state $\rho(0)$.
\end{theorem}
\begin{proof}
    Let $\rho(0)$ be a $4 \times 4$ Hermitian, positive semi-definite matrix with $\Tr(\rho(0)) = 1$. Without loss of generality, let's suppose
    \begin{equation*}
        \lambda(\rho(0))= \{\lambda_1, \lambda_2, \lambda_3, \lambda_4\}, \quad \lambda_1 \geq \lambda_2 \geq \lambda_3 \geq \lambda_4.
    \end{equation*}
    Let us define 
    \begin{equation}
         \mathcal{O}_\lambda = \{ U \rho(0) U^\dagger \mid U \in U(4) \}.
    \end{equation}
    Note that $\mathcal{O}_\lambda $ is  the set of all states unitarily equivalent to $\rho(0)$. As the continuous image of the compact, connected group $U(4)$, $\mathcal{O}_\lambda$ is itself a compact and connected set in the space of Hermitian matrices. 
    Let us define a real-valued function $f$ on this set:
    \[
    f(\rho) = \Tr((\sigma_z \otimes \sigma_z) \cdot \rho)
    \]
    Since the trace is a linear and therefore continuous function, $f$ is continuous on the compact, connected set $\mathcal{O}_\lambda$. By the Intermediate Value Theorem, the function $f$ must attain every value between its minimum and maximum on this set.
    \\
    \\
    The maximum value of $f(\rho)$ over the set $\mathcal{O}_\lambda$, $M$, is achieved by aligning the eigenvalues of $\rho$ and $\sigma_z \otimes \sigma_z$ in the same sorted order. That is, 
    \begin{align*}
        M = \max_{\rho \in \mathcal{O}_\lambda} f(\rho) &= (1)\lambda_1 + (1)\lambda_2 + (-1)\lambda_3 + (-1)\lambda_4 \\
        &= (\lambda_1 + \lambda_2) - (\lambda_3 + \lambda_4).
    \end{align*}
    And the minimum value $m$ is found by aligning them in opposite sorted orders. 
    \begin{align*}
        m = \min_{\rho \in \mathcal{O}_\lambda} f(\rho) &= (-1)\lambda_1 + (-1)\lambda_2 + (1)\lambda_3 + (1)\lambda_4 \\
        &= (\lambda_3 + \lambda_4) - (\lambda_1 + \lambda_2)
    \end{align*}
    Thus, the set of all attainable values for $f(\rho)$ is the continuous interval $[m, M]$. Now, note that based on the definition of the concurrence \cite{Wootters:1997id} $\mathcal{C} = \mathcal{C}(\rho(0)) = \text{max}\left\{0,\lambda_1-\lambda_2-\lambda_3-\lambda_4\right\}$, we have that 
    \begin{equation*}
        m \leq 0 \leq \mathcal{C} \leq M.
    \end{equation*}
    Therefore, there exists a state $\rho(t_f) \in \mathcal{O}_\lambda$ such that $f(\rho(t_f)) = \mathcal{C}$. This state is unitarily equivalent to $\rho(0)$ by definition of $\mathcal{O}_\lambda$.
\end{proof}

\begin{corollary}
    For any initial pure state $\rho(0)$ over a two qubit system, 
    the range of possible measurement values for $f(\rho) = \Tr( \sigma_z \otimes \sigma_z) \cdot \rho$ over the set of unitarily equivalent states $\mathcal{O}_{\lambda}$ is the interval $[-1,1]$.
\end{corollary}
\begin{proof}
    An initial state $\rho(0)$ is pure if and only if it is rank-1. Thus, the sorted eigenvalues are $\lambda = \{1,0,0,0\}$. Therefore, $M=1$ and $m=-1$. 
\end{proof}

\begin{corollary}
    For any initial rank-2 mixed state $\rho(0)$ over a two-qubit system, the range of possible measurement values for $f(\rho) = \Tr( \sigma_z \otimes \sigma_z) \cdot \rho$ over the set of unitarily equivalent states $\mathcal{O}_{\lambda}$ is the interval $[-1,1]$.
\end{corollary}

\begin{proof}
    An initial state $\rho(0)$ is rank-2 if and only if exactly two of its eigenvalues are non-zero. That is, $\lambda = \{\lambda_1, \lambda_2, 0, 0\}$ with $\lambda_1, \lambda_2 > 0 $ and $\lambda_1 + \lambda_2 = 1$. Therefore, $M=1$ and $m=-1$. 
\end{proof}

\section{Optimal control design}\label{Sec.Opt.Cont}

The general problem of optimal control \cite{bryson_1975} is to find a history of the control
vector $\mathbf{u}(t)$ for $0\leq t \leq t_f$ which forces the state from its initial value to its final value (along a trajectory) and, at the same time, maximizes or minimizes the cost functional. The resulting state history $\rho(t)$ for $0\leq t \leq t_f$ is an
optimal trajectory \cite{stengel_1994,STrimboli}. Along these lines, we follow reference \cite{behrman2008} to define the cost functional, which is a generalization of the Lagrangian used there. \\

Observe that we have proven the existence of both the final state and an optimal control but they are, so far, unknown. We plan to use a variational approach to find them. We will obtain a set of conditions for the minimization of a cost functional. These conditions also provide equations for the dynamics of the relevant variables. Their solution contains what we need. In order to use computer methods, we introduce a time discretization. The numerical solution of the discrete equations will be found by a gradient method over the control variables that improve iteratively a random initial guess. Our particular solution assumes that the control is piecewise constant, a sort of ``bang-bang'' control. This process is similar to what is known in quantum control as GRAPE/Krotov-style forward–backward gradient procedure \cite{Koch:2019ryo},\cite{KHANEJA2005296},\cite{Krotov1999}.

\subsection{Definition of the control functional}

Let us begin with the dynamic constraint, the Liouville-Von Neumann equation:
\begin{equation*}
 \frac{\partial \rho}{\partial t}=\frac{1}{i\hbar}[H,\rho]
\end{equation*}
Very often, a scalar algebraic function of the final state is chosen as part of the cost function.
\begin{align*}
 \frac{1}{2}[d-\langle \mathcal O(t_f) \rangle]^2
\end{align*}
where $\mathcal{O}$ is an output measure 
or some function of this measure \cite{behrman2008}.  For $d$ we will use the concurrence \cite{Wootters:1997id} of the initial state  and evaluate the quadratic error of some measurement of the final state. The cost function is chosen as:
\begin{align}
 {\mathcal J}= & \frac{1}{2}[d-\langle \mathcal O \rangle]^2+\int_{t_0}^{t_f}\,Tr\left(\lambda^{\dagger}\left(\frac{1}{i\hbar}[H,\rho]-\frac{\partial\rho}{\partial t}\right)\right)\,dt \nonumber\\
 &+\frac{1}{8}\int_{t_0}^{t_f}\,Tr\left(H^{\dagger}H\right)\,dt
\end{align}
where $Tr$ is the trace. The second term uses Lagrange multipliers to ensure that the dynamics of the system follow the Liouville-Von Neumann equation, while the last term is a ``cost of control''. The idea behind the mathematical expression is to minimize the number of interventions necessary during the time evolution of the system. The $1/8$ is an arbitrary normalization factor that defines the relative size of the contribution of this term. This equation can be rewritten in terms of the components of each matrix:
\begin{align*}
{\mathcal J}= & \frac{1}{2}[d-\langle \mathcal O(t_f)\rangle]^2+\int_{t_0}^{t_f}\,\sum_{i,j}\lambda^{*}_{ji}\left(\frac{1}{i\hbar}[H,\rho]-\frac{\partial\rho}{\partial t}\right)_{ij}dt \\ &+\frac{1}{8}\int_{t_0}^{t_f}\,\sum_{i,j}H^{*}_{ji}H_{ij}\,dt
\end{align*}
$\lambda$ is a matrix containing the Lagrange multipliers. Observe that each component of the density matrix is considered an independent variable.

We can rewrite the cost function introducing the {\it control Hamiltonian}:
\begin{align*}
 {\mathcal J}= & \frac{1}{2}(d-\langle \mathcal O(t_f)\rangle)^2 \\ &+\bigintsss_{t_0}^{t_f}\,\left\{{\mathcal H}(\rho(t),{\bf u(t)},\lambda(t))-Tr\left(\lambda^{\dagger}\frac{\partial\rho}{\partial t}\right)\right\}dt
\end{align*}

where
\begin{equation*}
 {\mathcal H}(\rho(t),{\bf u}(t),\lambda(t))\equiv\frac{1}{i\hbar}\Tr\left(\lambda^\dagger[H,\rho]\right)+\frac{1}{8}Tr\left(H^{\dagger}H\right)
\end{equation*}

Integrating by parts, we obtain the final form of the cost functional:
\begin{align}
 {\mathcal J}= & \frac{1}{2}(d-\langle \mathcal O(t_f)\rangle)^2\nonumber \\
 & +\bigintsss_{t_0}^{t_f}\,\left\{{\mathcal H}(\rho(t),{\bf u}(t),\lambda(t))+Tr\left(\rho\frac{d\lambda^\dagger}{dt}\right)\right\}dt\nonumber \\
  &-Tr(\rho(t_f)\lambda^\dagger(t_f))
\end{align}

\subsection{Equations of motion for the control variables}

Now, let us find the Euler-Lagrange equations. Varying over the corresponding variables:
\begin{align*}
 0= & ~\delta {\mathcal J}= -(d-\phi(\rho_f))\frac{\partial\phi}{\partial\rho_f}\Delta \rho_f \\ &+\bigintsss_{t_0}^{t_f}\,\left\{\frac{\partial\mathcal{H}}{\partial \bf u}\Delta{\bf u}+\frac{\partial\mathcal{H}}{\partial\rho}\Delta{\rho}+\frac{\partial}{\partial \rho}Tr\left(\frac{\partial\lambda^\dagger}{\partial t}\rho\right)\Delta\rho\right\}dt \\ &-\frac{\partial}{\partial\rho_f}Tr(\lambda^\dagger_f\rho_f)\Delta\rho_f
\end{align*}
We have used $\frac{\partial}{\partial A}\Tr(AB)=B^T$ (see Appendix  \ref{Ap:mathresults}) and defined $\phi(\rho_f)=\langle \mathcal O(t_f)\rangle$. Since all variations are independent, we can write:
 \begin{enumerate}[(i)]
  \item\label{eq:el_1} ${\displaystyle\frac{\partial\mathcal{H}}{\partial \bf u}=0}$
  \item\label{eq:el_2}  ${\displaystyle\frac{\partial \mathcal{H}}{\partial\rho}+\frac{\partial\bar{\lambda}}{\partial t}=0}$
  \item\label{eq:el_3} ${\displaystyle-(d-\phi(\rho_f))\frac{\partial\phi}{\partial\rho_f}-\bar{\lambda}_f=0}$
 \end{enumerate}
  Let us start with \ref{eq:el_1}:
\begin{align*}
 \frac{\partial\mathcal{H}}{\partial \bf u}&=\frac{1}{i\hbar}\Tr\left(\lambda^\dagger\left[\frac{\partial H}{\partial\bf u},\rho\right]\right)+\frac{1}{8}Tr\left(\frac{\partial H^{\dagger}}{\partial \bf u}H+H^{\dagger}\frac{\partial H}{\partial \bf u}\right)
 \intertext{Using the fact that $H$ is real and symmetric (see Appendix  \ref{Ap:Hamiltonian}):}
 &=\frac{1}{i\hbar}\Tr\left(\lambda^\dagger\left[\frac{\partial H}{\partial\bf u},\rho\right]\right)+\frac{1}{4}Tr\left(\frac{\partial H}{\partial \bf u}H\right)
 \intertext{Now, observe that:}
 u_l&=\frac{1}{4}Tr\left(\frac{\partial H}{\partial u_l}H\right)
 \intertext{And therefore:}
 \frac{\partial\mathcal{H}}{\partial u_l}&=\frac{1}{i\hbar}\Tr\left(\lambda^\dagger\left[\frac{\partial H}{\partial u_l},\rho\right]\right)+u_l
\end{align*}

Given the structure of the Hamiltonian, we have six equations:
\begin{align*}
 0&=\frac{\partial \mathcal{H} }{\partial \kappa_A}=\kappa_A+\frac{1}{i\hbar}\Tr\left(\lambda^\dagger\left[\sigma_x\otimes\mathds{1},\rho\right]\right) \\ 0&=\frac{\partial \mathcal{H} }{\partial \kappa_B}=\kappa_B+\frac{1}{i\hbar}\Tr\left(\lambda^\dagger\left[\mathds{1}\otimes\sigma_x,\rho\right]\right)\\
   0&=\frac{\partial \mathcal{H} }{\partial \epsilon_A}=\epsilon_A+\frac{1}{i\hbar}\Tr\left(\lambda^\dagger\left[\sigma_z\otimes\mathds{1},\rho\right]\right) \\ 0& =\frac{\partial \mathcal{H} }{\partial \epsilon_B}=\epsilon_B+\frac{1}{i\hbar}\Tr\left(\lambda^\dagger\left[\mathds{1}\otimes\sigma_z,\rho\right]\right)\\
  0&=\frac{\partial \mathcal{H} }{\partial \zeta}=\zeta+\frac{1}{i\hbar}\Tr\left(\lambda^\dagger\left[\sigma_z\otimes\sigma_z,\rho\right]\right) \\ 0&=\frac{\partial \mathcal{H} }{\partial \nu}=\nu+\frac{1}{i\hbar}\Tr\left(\lambda^\dagger\left[\sigma_x\otimes\sigma_x,\rho\right]\right)\\
\end{align*}

Let's continue with \ref{eq:el_2}. Evaluating the derivative of the first term:
\begin{align}
 \frac{\partial\mathcal{H}}{\partial \rho}&=\frac{1}{i\hbar}[H,\bar{\lambda}]\nonumber
\intertext{thus:}
 \frac{\partial\bar{\lambda}}{\partial t}&=-\frac{1}{i\hbar}[H,\bar{\lambda}]\nonumber
 \intertext{Taking the complex conjugate on both sides, we finally obtain the dynamic equation for the Lagrange multipliers:}
  \frac{\partial\lambda}{\partial t}&=\frac{1}{i\hbar}[H,\lambda]
\end{align}
We can work now with the third Euler-Lagrange equation \ref{eq:el_3}:
\begin{align}
 \bar{\lambda}_f&=-(d-\phi(\rho_f))\frac{\partial\phi}{\partial\rho_f}\nonumber
 \intertext{As observable, we choose $\sigma_z\otimes\sigma_z$:}
 \phi&=Tr(\rho_f(\sigma_z\otimes\sigma_z))
 \intertext{So:}
 \lambda_f&=-(d-\phi(\rho_f))(\sigma_z\otimes\sigma_z)
\end{align}
We have in our hands a two-point boundary value problem \cite{stengel_1994}, as the constants of integration for the state and the Lagrange multipliers are specified at opposite ends of the boundary interval:
\begin{align}
 \rho(t)&=\mathcal{U}(t)\,\rho_i\,\mathcal{U}^\dagger(t)\\
 \lambda(t)&=\mathcal{U}^\dagger(t)\,\lambda_f\,\mathcal{U}(t)\\
 \intertext{Observe that if $\lambda_f$ is Hermitian, so is $\lambda$, which is our case. The corresponding solutions are well known from quantum mechanics:}
 \mathcal{U}(t)&=\exp_{+}\left[-\frac{i}{\hbar}\int_{t_0}^t d\tau\,H(\tau)\right]\nonumber
 \intertext{If the Hamiltonian is time independent, the propagator is:}
 \mathcal{U}(t)&=\exp\left[-\frac{i}{\hbar}H(t-t_0)\right]\nonumber
\end{align}

In the following section, we will discretize the time variable and apply the control process to a sequence of time intervals.

\section{Quantum control algorithm} \label{Sec.Quant.Cont.Alg}

We will now extend the definition to a sequence of N-time slices where the Hamiltonian of the system is constant:
\begin{align}
 {\mathcal J}&=\frac{1}{2}[d-\phi(\rho_N)]^2 \nonumber\\
 &+\sum_{k=1}^{N-1}\,\left\{{\mathcal H}(\rho_k,{\bf u}_k,\lambda_k)+Tr\left(\rho_k\frac{d\lambda_k^\dagger}{dt}\right)\right\}\nonumber \\
 &-Tr(\rho_N\lambda^\dagger_N)
 \end{align}
where:
 \begin{align}
 {\mathcal H}(\rho_k,{\bf u}_k,\lambda_k)&=\frac{1}{8}Tr(H^\dagger({\bf u}_k)H({\bf u}_k))\nonumber \\
    &+\frac{1}{i\hbar}Tr(\lambda^\dagger_{k+1}[H({\bf u}_k),\rho_k])
\end{align}
Following a procedure similar to the one outlined above and using the fact that $\lambda$ is Hermitian, we obtain the following Euler-Lagrange equations:
 \begin{enumerate}[(i)]
  \item\label{eq:el_1d} $ {\displaystyle\frac{\partial\mathcal{H}}{\partial u_{l,k}}=\frac{1}{i\hbar}Tr\left(\lambda_{k+1}\left[\frac{\partial H(u_{l,k})}{\partial u_{l,k}},\rho\right]\right)+u_{l,k}}$
  \item\label{eq:el_2d}  ${\displaystyle \frac{\partial\lambda_k}{\partial t}=\frac{1}{i\hbar}[H({\bf u}_k),\lambda_{k+1}]}$
  \item\label{eq:el_3d} ${\displaystyle-(d-Tr(\rho_N\sigma_{zz}))\sigma_{zz}-\lambda_N=0}$
 \end{enumerate}
 We need to remember to supplement these relations with the time evolution of the state variables:
 \begin{equation*}
   \frac{\partial\rho_k}{\partial t}=\frac{1}{i\hbar}[H({\bf u}_k),\rho_{k}]
 \end{equation*}

 The above relations allow us to find each ${\bf u}_k$. So we can find a control path for any given initial state. This is equivalent to saying that we can unitarily transform this initial state in such a way that the $z$-projection is equal to the concurrence.  That this transformation exists was the main result of section II above.  
 
 Below we present a gradient method to solve the above equations for the control, $u$  (quantum parameters) time history.  Although the quantum dynamics are not control affine in the traditional linear optimal control sense, it might be possible to assume a linear relationship between $\lambda$ and $\rho$ that would enable a closed form solution for $u$ via combining equations (i) and (ii).  We leave that for a future paper. 

 We solve the two-point boundary value problem using the following gradient method \cite{STrimboli,bryson_1975}:
\begin{enumerate}
 \item Initialize all of the control variables $\{{\bf u}_k;\,\,k=1,\cdots N-1\}$ (randomly or with a priori knowledge)
 \item Compute the state variables $\rho_k$ using:
 \begin{equation*}
  \rho_{k+1}=e^{-\frac{i\Delta t}{\hbar}H({\bf u}_{k+1})}\rho_{k}e^{\frac{i\Delta t}{\hbar}H({\bf u}_{k+1})}\,\,\,,\,k=0,\cdots, N-1
 \end{equation*}
 with initial condition $\rho_0$.
\item Compute the adjoint variables $\lambda_k$ using:
 \begin{equation*}
  \lambda_k=e^{\frac{i\Delta t}{\hbar}H({\bf u}_k)}\lambda_{k+1}e^{\frac{-i\Delta t}{\hbar}H({\bf u}_k)}\,\,\,,\,k=N-1,\cdots,1
 \end{equation*}
 with final condition $\lambda_N$.
 \item Stop when $\frac{\partial\mathcal{H}}{\partial {\bf u}_{k}}\approx0$. In terms of the components:
 \begin{align*}
      \sum_{l}\sum_{k=1}^{N-1}\left\{\frac{1}{i\hbar}\Tr\left(\lambda_{k+1}\left[\frac{\partial H({u}_{l,k})}{\partial {\bf u}_{k}},\rho_k\right]\right)+{u}_{l,k}\right\}&<\epsilon
  \end{align*}
 \item If the stopping criterion is not satisfied, improve {\bf u} with the update rule in the equation. This is done by setting:
 \begin{align*}
      {\bf u}_{k,new}={\bf u}_{k,old}-\eta\frac{\partial \mathcal{H}({\bf u}_{k})}{\partial {\bf u}_{k}}
 \end{align*}
  where $\eta$ is known as the learning rate.
\end{enumerate}
This last step is justified by the decrease in the cost functional:
 \begin{enumerate}[(i)]
  \item If $ {\displaystyle \frac{\partial\mathcal{H}}{\partial {\bf u}_{k}}>0}$, then $d{\bf u}<0$ will produce $\delta {\mathcal J}<0$ and hence ${\mathcal J}$ will decrease.
  \item If $ {\displaystyle \frac{\partial\mathcal{H}}{\partial {\bf u}_{k}}<0}$, then $d{\bf u}>0$ will produce $\delta {\mathcal J}<0$ and hence ${\mathcal J}$ will also decrease.
 \end{enumerate}
 
The algorithm finds both the unitary matrix that performs the transformation and the final density matrix.
For the unitary matrix, we use, in general, four different sets of parameters, see eq.(\ref{eq:Hamiltonian}), which implies four different Hamiltonians acting in sequence to drive the time evolution of the initial density matrix.
Theoretically, there is only a need for two sets of parameters.  Their existence is guaranteed by the following theorem:
\begin{theorem}[Theorem 3.2.7 of \cite{Dalessandro_2022}]
 A generic pair of Hamiltonians $A,B\in su(n)$ generates $su(n)$
\end{theorem}

\section{Numerical implementation}\label{Sec.Dis.Resul}

We implemented the above algorithm in MATLAB/Octave. As it has already been stated, the algorithm takes an initial state represented by a density matrix $\rho(0)$, and returns a well-designed final state $\rho(t_f)$ and a set of controls associated with the transformation. This set of values is known as the control path. For most of the  density matrices, four time steps, and so four different Hamiltonians, complete the job of finding the unitary matrix that transforms the initial state.

From the control variables in each time slice $\{{\bf u}_k;\,\,k=1,\cdots N-1\}$ we can build a unitary matrix that represents the control process during this time interval. The sequential product of these matrices is equal to the unitary matrix $\mathcal{U}$ responsible for the transformation between the initial and final states.

We tested the software on samples containing density matrices parameterized by the method described in \cite{Kong:2021faq}. After sampling uniformly the parameter space, we produced sets containing between 100 to 1000 different matrices with pure and mixed states in them. Typical results can be observed in Fig.\ref{fig:concurrences} where we compare the concurrence of the initial states with the projection on the $\sigma_z\otimes\sigma_z$ direction of the corresponding final state. The red line corresponds to the $x=y$ function.
\begin{figure}
\includegraphics[width=\columnwidth]{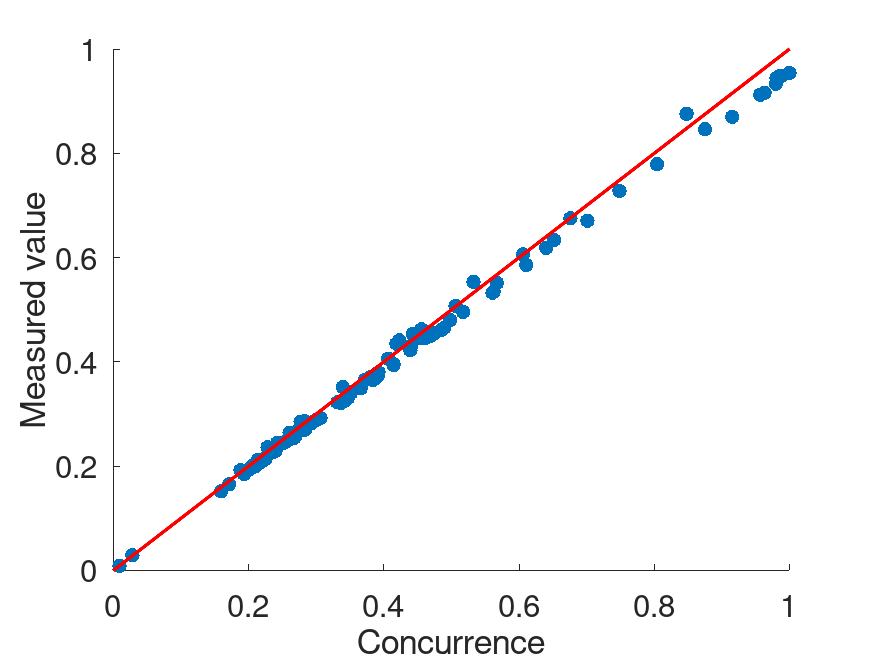}
\caption{Scatter plot for the concurrence of the initial state and the $\sigma_z\otimes\sigma_z$ measurement on the final state compared with the line $x=y$.} 

\label{fig:concurrences}
\end{figure}

 To numerically compare between the concurrence and the final state measurement we use the relative error defined as $\tfrac{|\mathcal{C}(\rho(0))-Tr(\sigma_z\otimes\sigma_z\rho(t_f))|}{\mathcal{C}(\rho(0))}$.  In Fig.\ref{fig:relative_error} we plot a histogram of this error for the same sample of Fig.\ref{fig:concurrences}. \\
\begin{figure}
\includegraphics[width=\columnwidth]{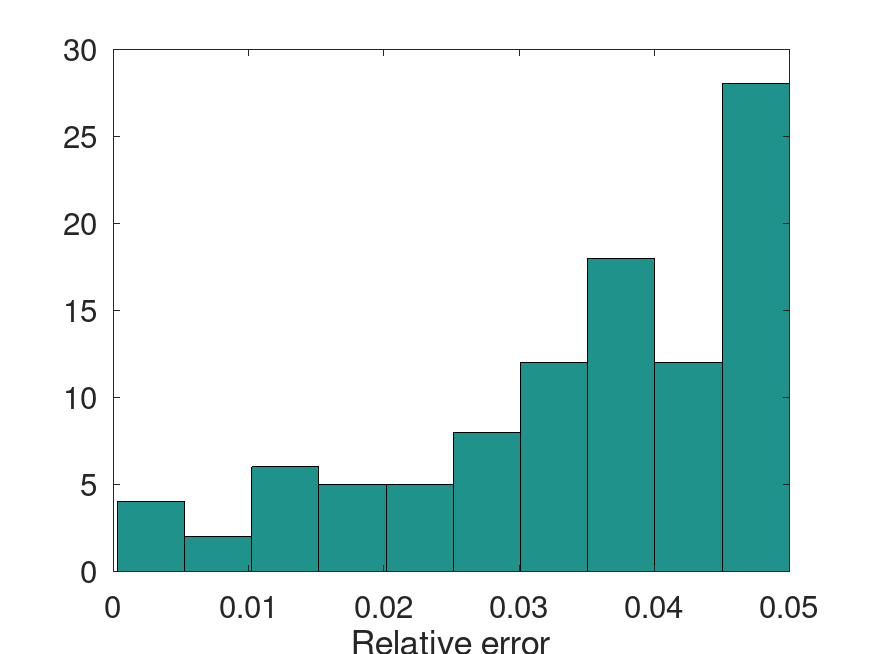}
\caption{Histogram for the relative error between the concurrence of the initial state and the $\sigma_z\otimes\sigma_z$ measurement on the final state.}
 \label{fig:relative_error}
\end{figure}

It can be observed that the algorithm delivers a solution for the final state density matrix with its projection inside an interval defined around the concurrence of the initial state and the chosen maximum error. For our computer runs, we selected $5\%$. This choice reflects  a combination of a reasonable running time and numerical accuracy for our results.

\section{Conclusions}\label{thatsallfolks}

In this paper, we have presented an efficient verification protocol for two-qubit entanglement generation.
We proved that an initial state of a two-qubit system can be evolved through a unitary transformation into a designated final state, which we can optimally design to yield entanglement with a $ZZ$ measurement, or equivalently, local qubit $Z$ measurements. Our numerical simulation results demonstrated that a small ansatz is sufficient to implement the unitary transformation through an optimal control process given a specific error tolerance.
Furthermore, these results provide the mathematical justification for the 2008 results \cite{behrman2008} by some of the authors.
Note that our approach is inherently robust to constant parameter drift as shown by Corollary \ref{cor:robustness}. This is also in line with the previous work \cite{noiseNguyen} by some of the authors.

While our results here apply only to a two-qubit system, much of this framework will generalize to an N-qubit system under the  $Z^{\otimes N}$ or even general measure. Our work in quantum transfer learning \cite{qtransfer} provides some support for wider applicability, as well. Moreover, while we show here specifically an entanglement verification, our results can be generalized to the measurement or verification of other quantities.

\begin{acknowledgments}
Ricardo Rodriguez is grateful for the Watkins Summer Fellowship at Wichita State University and the 2024 Visiting Faculty Program with Fermi National Accelerator Laboratory, during which much of his research for this project was performed. He is also grateful to Nathan Thompson for the very valuable conversations and support at Wichita State University. 
This manuscript has been authored by FermiForward Discovery Group, LLC under Contract No. 89243024CSC000002 with the U.S. Department of Energy, Office of Science, Office of High Energy Physics. We would like also to thank our anonymous referees who helped us to make clearer our discussion and point out new way to frame our results.
\end{acknowledgments}

\appendix
\appendixpage

We have here a few results related to the main discussion.

\section{Some Mathematical Results}\label{Ap:mathresults}

How to find the derivative of a trace.\\

Let us find $\frac{\partial}{\partial A}\Tr[AB]$:
\begin{align*}
 \frac{\partial}{\partial A_{lm}}\Tr(AB)&=\frac{\partial}{\partial A_{lm}}\sum_{i,j}A_{ij}B_{ji}\\
 &=\sum_{ij}B_{ji}\delta_{il}\delta_{mj}\\
 &=B_{ml}
 \intertext{Therefore:}
 \frac{\partial}{\partial A}\Tr(AB)&=B^T
\end{align*}

We have been using standard techniques in optimal control which assume that that the control functional is real. In the theorem below we show this fact.

\begin{theorem}
$\mathcal{J}$ is a real.

\end{theorem}

\begin{proof}
We have been implicitly using that $\mathcal{J}$ is a real function. To prove this fact, we need to show that
\begin{equation*}
 \mathcal{S}\equiv\sum_{i,j}\lambda^{*}_{ji}\left(\frac{1}{i\hbar}[H,\rho]-\frac{\partial \rho}{\partial t}\right)_{ij}
\end{equation*}
is also real, the other terms in $\mathcal{J}$ are clearly real. This is equivalent to showing that $\mathcal{S}=\mathcal{S}^{*}$. Let us take the complex conjugate of $\mathcal{S}$:
\begin{align*}
 \mathcal{S}^{*}&=\left\{\sum_{i,j}\lambda^{*}_{ji}\left(\frac{1}{i\hbar}[H,\rho]-\frac{\partial \rho}{\partial t}\right)_{ij}\right\}^*\\
    &= \sum_{i,j}\lambda_{ji}\left(\frac{-1}{i\hbar}[H,\rho]_{ij}^*-\frac{\partial \rho_{ij}^*}{\partial t}\right)\\
     &= \sum_{i,j}\lambda_{ji}\left(\frac{-1}{i\hbar}\sum_{l}(H_{il}^*\rho_{lj}^*-\rho_{il}^*H_{lj}^*)-\frac{\partial \rho_{ij}^*}{\partial t}\right)\\
 \intertext{By hermiticity $a_{ij}^*=a_{ji}$, thus:}
     &= \sum_{i,j}\lambda_{ij}^{*}\left(\frac{-1}{i\hbar}\sum_{l}(H_{li}\rho_{jl}-\rho_{li}H_{jl})-\frac{\partial \rho_{ji}}{\partial t}\right)\\
     &= \sum_{i,j}\lambda_{ij}^{*}\left(\frac{1}{i\hbar}\sum_{l}(H_{jl}\rho_{li}-\rho_{jl}H_{li})-\frac{\partial \rho_{ji}}{\partial t}\right)\\
     \mathcal{S}^{*}&=\sum_{i,j}\lambda_{ij}^{*}\left(\frac{1}{i\hbar}[H,\rho]-\frac{\partial \rho}{\partial t}\right)_{ji}
\end{align*}
therefore $\mathcal{S}^{*}=\mathcal{S}$.
\end{proof}

\section{Dynamic Hamiltonian for two qubits}\label{Ap:Hamiltonian}
The physical Hamiltonian used is a generalization of \cite{behrman2008}:
\begin{align*}
H= & \, \kappa_A \sigma_x\otimes\mathds{1}+\kappa_B \mathds{1}\otimes\sigma_x+\epsilon_A\sigma_z\otimes\mathds{1}+\epsilon_B \mathds{1}\otimes\sigma_z \\ &+\zeta\sigma_z\otimes\sigma_z+\nu \sigma_x\otimes\sigma_x\\
= &
\begin{pmatrix}
  \epsilon_{+}+\zeta&\kappa_B & \kappa_A & \nu\\
  \kappa_B & \epsilon_{-}-\zeta& \nu & \kappa_A \\
  \kappa_A & \nu & -\epsilon_{-}-\zeta & \kappa_B \\
  \nu & \kappa_A & \kappa_B & -\epsilon_{+}+\zeta
 \end{pmatrix}
\end{align*}
where $\epsilon_{\pm} = \epsilon_A \pm \epsilon_B$.

\section{Definition of Concurrence}\label{Ap:concurrence}


This section follows \cite{Mintert:2005207}. The original definition of concurrence \cite{PhysRevA.54.3824,PhysRevLett.78.5022} for a bipartite two-level system is given in terms of a special basis:
\begin{equation*}
 |e_1\rangle=|\Phi^+\rangle,\,\,\,\,|e_2\rangle=|\Phi^-\rangle\ , \,\,|e_3\rangle=i|\Psi^+\rangle,\,\,\,|e_4\rangle=|\Psi^-\rangle\,\,\,\
\end{equation*}
where $|\Phi^\pm\rangle=(|00\rangle\pm|11\rangle)/\sqrt{2}$ and $|\Psi^\pm\rangle=(|01\rangle\pm|10\rangle)/\sqrt{2}$ are the {\it Bell} states. Using this particular basis, the concurrence $\mathcal{C}$ of a pure state $|\Psi\rangle$ is defined as
\begin{equation*}
 \mathcal{C}(\Psi)=\left|\sum_i\langle e_i|\Psi\rangle^2\right|
\end{equation*}
Writing this definition more explicitly, $\mathcal{C}(\Psi)=\left|\sum_i\langle \Psi^*|e_i^*\rangle\langle e_i|\Psi\rangle\right|$ one can obtain after summing over the basis states,
\begin{equation*}
 \mathcal{C}(\Psi)=\left|\langle \Psi^*|\sigma_y\otimes\sigma_y|\Psi\rangle\right|
\end{equation*}
Note that $\Psi^*$ is the complex conjugate of $\Psi$ with the conjugation performed in the standard basis $\{|00\rangle,|01\rangle,|10\rangle,|11\rangle\}$. We have to be careful with the above equation. The inner product uses the transpose of the complex conjugate of the state, so the adjoint of $\Psi^*$ is just the transpose of $\Psi$. Therefore $\langle \Psi^*|=\sum_{i,j}\Psi_{ij}\langle ij|$ is the transpose and not the adjoint of $|\Psi\rangle=\sum_{i,j}\Psi_{ij}|ij\rangle $. The above definition is the most commonly used. \\

\begin{definition}[Concurrence for Mixed States]
The concurrence of a mixed state is given by using the convex roof construction:
\begin{align*}   
   \mathcal{C}(\rho)=&\inf_{\{p_i,\Psi_i\}}\sum_i\,p_i\mathcal{C}(\Psi), \,\, 
 \text{with } \\ & p_i>0,\,~ \sum_i p_i = 1, \text{  and }\,\,\,  \\ &\rho=\sum_i\,p_i|\Psi_i\rangle\langle\Psi_i|
\end{align*} 
\end{definition}
An explicit formula for this infimum is known \cite{Wootters:1997id}:
\begin{equation*}
 \mathcal{C}(\rho)=\text{max}\left\{0,\lambda_1-\lambda_2-\lambda_3-\lambda_4\right\}
\end{equation*}
where the $\lambda_i$ are the eigenvalues, in decreasing order, of the Hermitian matrix:
\begin{equation*}
 R\equiv\sqrt{\sqrt{\rho}\tilde{\rho}\sqrt{\rho}}
\end{equation*}
with
\begin{equation*}
 \tilde{\rho}=(\sigma_y\otimes\sigma_y)\rho^*(\sigma_y\otimes\sigma_y)
\end{equation*}
and $\rho^*$ is the complex conjugate of $\rho$.

This definition is equivalent to \cite{Wootters:1997id}:
 \begin{equation*}
\mathcal{C}(\rho)\equiv \text{max}(0,\sqrt{\lambda_1}-\sqrt{\lambda_2}-\sqrt{\lambda_3}-\sqrt{\lambda_4})
\end{equation*}
\\
\noindent in which $\lambda_1,\lambda_2,\lambda_3,\lambda_4$ are eigenvalues in decreasing order of the non-Hermitian matrix
$\rho\tilde{\rho}$.\\

Observe that the concurrence is a positive number such that $0\leq \mathcal{C}(\rho) \leq 1 $.\\

\addcontentsline{toc}{section}{References}
\bibliography{introduction_optimization}

\end{document}